\documentclass[10 pt,final,journal,letterpaper,oneside,twocolumn]{IEEEtran}
\usepackage{multicol}   
\usepackage{multirow}
\usepackage[pdftex]{graphicx}
\usepackage[ruled]{algorithm2e}
\usepackage{algorithmic}
\usepackage[small]{caption}
\usepackage{float}
\usepackage{amssymb,amsmath}
\usepackage{graphicx}
\usepackage{subfigure}
\usepackage{xspace}
\usepackage{amsthm}
\usepackage{fancyhdr}
\newtheorem{prop}{Proposition}
\usepackage{caption,setspace}
\usepackage{amssymb,amsmath}
\usepackage{textcomp} 
\usepackage{epstopdf}
\usepackage{bbding}
\usepackage{cite} 
\usepackage{color}
\usepackage{pifont}
\begin{document}
%
\title{\huge Backscatter Sensors Communication for 6G Low-powered NOMA-enabled IoT Networks under Imperfect SIC}
\author{Manzoor Ahmed, Wali Ullah Khan, Asim Ihsan, Xingwang Li, Jianbo Li, and Theodoros A. Tsiftsis\thanks{Manzoor Ahmed and J. Li are with the College of Computer Science and Technology, Qingdao University, Qingdao 266071, China. (emails: manzoor.achakzai@gmail.com, lijianbo@qdu.edu.cn).

Wali Ullah Khan is with the Interdisciplinary Centre for Security, Reliability and Trust (SnT), University of Luxembourg, 1855 Luxembourg City, Luxembourg (Emails: waliullah.khan@uni.lu, waliullahkhan30@gmail.com).

Asim Ihsan is with the with Department of Information and Communication Engineering, Shanghai Jiao Tong University, Shanghai 200240, China. (email:
ihsanasim@sjtu.edu.cn).
Xingwang Li is with the School of Physics and Electronic Information Engineering, Henan Polytechnic University, Jiaozuo, China (email: lixingwangbupt@gmail.com).

Theodoros A. Tsiftsis is with School of Electrical and Information Engineering, Jinan University, Zhuhai
519070, China (email: theodoros.tsiftsis@gmail.com).

}\vspace{-0.6cm}}%

\markboth{Submitted to IEEE}%
{Shell \MakeLowercase{\textit{et al.}}: Bare Demo of IEEEtran.cls for IEEE Journals} 

\maketitle

\begin{abstract}
The combination of non-orthogonal multiple access (NOMA) using power-domain with backscatter sensor communication (BSC) is expected to connect a large-scale Internet of things (IoT) devices in future sixth-generation (6G) era. In this paper, we introduce a BSC in multi-cell IoT network, where a source in each cell transmits superimposed signal to its associated IoT devices using NOMA. The backscatter sensor tag (BST) also transmit data towards IoT devices by reflecting and modulating the superimposed signal of the source. A new optimization framework is provided that simultaneously optimizes the total power of each source, power allocation coefficient of IoT devices and reflection coefficient of BST under imperfect successive interference cancellation decoding. The objective of this work is to maximize the total energy efficiency of IoT network subject to quality of services of each IoT device. The problem is first transformed using the Dinkelbach method and then decoupled into two subproblems. The Karush-Kuhn-Tucker conditions and Lagrangian dual method are employed to obtain the efficient solutions. In addition, we also present the conventional NOMA network without BSC as a benchmark framework. Simulation results unveil the advantage of our considered NOMA BSC networks over the conventional NOMA network.
\end{abstract}


\IEEEpeerreviewmaketitle

\section{Introduction}
In the last couple of years, Internet-of-things (IoT) has been emerged as a new technological innovation in a wide range of applications such as smart factories, smart cities, smart homes, smart hospitals, autonomous vehicles, and so on \cite{9261963,liu2019next}. The IoT is expected to connect billions of sensor devices in the future sixth-generation (6G) systems \cite{9468352}, which would require the efficient utilization of existing spectrum resources \cite{9516696,2021726}. However, one of the key challenges would be energy issues especially for those systems where the battery replacement of sensor devices can be very costly \cite{8933559}. In particular, the sensor devices  which are hidden in walls and appliances or deployed in radioactive areas and pressurized pipes, making battery replacement difficult if not possible \cite{8861078}. In such circumstances, ambient energy harvesting is a highly desirable approach to maintain the life of sensor devices for a long period \cite{8253544,khan2021energy}. It is important to mention here that ambient energy can sufficiently power sensor devices due to their low energy consumption. In this regard, a promising solution is Backscatter communication (BC) \cite{ihsan2021energy}. BC allows the sensor devices to transmit data by reflecting and modulating the existed radio frequency signal \cite{8368232}.  
\subsection{Technical Literature Review}
Recently, power-domain non-orthogonal multiple access (NOMA) has gained significant importance due to its high spectral efficiency and massive connectivity \cite{9154358,9479745}. Compared to orthogonal multiple access (OMA) techniques, NOMA supports multiple IoT devices over the same spectrum/time resources which can be accomplished through two techniques, i.e., superposition coding at transmitter side and successive interference cancellation (SIC) at receiving side \cite{7842433,ali2022fair}. Various research works on backscatter communication in traditional OMA networks have been studied in literature. For example, Guo {\em et al.} \cite{8692391} have provided the efficient power allocation approach for cooperative BC to investigate the achievable rate of the system. The authors of \cite{9051982} computed a closed-form solution for the outage probability (OP) of BC. In \cite{jameel2019simultaneous}, the authors derived a closed-form expression for the OP of a BC system over Rayleigh fading channels. They also investigated the trade-off between harvested energy and data rate through power splitting factor. Qian {\em et al.} \cite{8423609} calculated a closed-form expression for the symbol-error rate and designed an efficient multi-level energy detector for BC system. The authors of \cite{8093703} investigated an optimization problem for throughput maximization of BC through calculating the optimal reflection coefficient (RC) and the trade-off between active and sleep state. Jameel {\em et al.} \cite{9129364} exploited Q-learning approach to improve the achievable data rate while the constraint on delay is taken into account. In addition, Li {\em et al.} \cite{9363336} investigated security and reliability of BC by through calculating the OP and intercept probability (IP) of the system. Recently, the performance of BC has been investigated using reinforcement learning techniques. The authors of \cite{9024401,9162720} have provided intelligent power allocation algorithms to improve the performance of BC systems. 

The integration of BC in NOMA has recently been studied in literature. For example, in \cite{le2019outage}, the expression of OP has been derived in NOMA BC network where a source is equipped with multi-antenna scenario. Zhang {\em et al.} \cite{8636518} have derived a closed-form expression for the OP and ergodic capacity in NOMA BC symbiotic radio systems. The work of \cite{9131891} has studied the security issues of NOMA BC network. Khan {\em et al.} \cite{9345447} have considered NOMA BC in vehicle-to-everything network to maximize the sum capacity of the system. To improve the spectrum management and network capacity, Liao {\em et al.} \cite{8962090} have studied resource allocation problem in full duplex NOMA BC networks. The work of \cite{8439079} has improved the average successful decoding bit by efficient RC selection criteria in NOMA BC network. In similar study, Farajzadeh {\em et al.} \cite{8761125} have optimized unmanned aerial vehicle altitude and maximize the successful decoded bit rate of NOMA BC network. Yang {\em et al.} \cite{8851217} have optimized the time and RC of BC to maximize the system minimum throughput. Moreover, the work of \cite{8877102} has investigated the OP and the throughput of NOMA BC system. Besides, Li {\em et al.} \cite{li2019secure} have studied the physical layer security of multiple-input single-output NOMA BC network. In \cite{9223730}, the authors have proposed an optimization problem of transmit power and RC for BC to maximize the energy efficiency (EE) of the system. To investigate the security and reliability of NOMA BC system, the authors of \cite{9319204} have investigated the OP and IP under channel estimation error, imperfect SIC and residual hardware impairment. In addition, the joint optimization of power and RC under imperfect SIC was solved in \cite{9328505} to maximize the sum rate of NOMA BC system.

\subsection{Motivation and Contributions}
The above-existed literature \cite{8692391,9051982,jameel2019simultaneous,8423609,8093703,9129364,9363336,9024401,9162720,le2019outage,8636518,9131891,9345447,8962090,8439079,8761125,8851217,8877102,li2019secure,9223730} considers perfect SIC at the receiver side which is impractical in real systems. Of course, a decoding error can occur during the SIC process at receiver side such that the interference of other devices cannot be removed. This will result in significant degradation of the system performance. Besides that, most of the research works consider only single-cell and two-user scenarios. Generally, a network is consists of different cells having various sizes. These cells normally share the same spectrum resources to enhance the spectral efficiency, result in causing inter-cell interference to each other. Moreover, the works in \cite{9319204,9328505} consider imperfect SIC in the single-cell system but their objectives were to improve the sum-capacity and physical layer security. Based on the above observations, there is a need to investigate a system performance with multi-cell, considering inter-cell interference and imperfect SIC decoding.  Thus, the problem that jointly optimizes the total power budget of source, power allocation coefficient (PAC) of IoT devices, and RC of backscatter tag in each cell to investigate the EE of NOMA BC in multi-cell network under imperfect SIC has not yet been investigated, to the best of our knowledge. To bridge this gap, this work aims at proposing a new optimization approach for maximizing the system EE of the multi-cell NOMA backscatter sensors communication (BSC) network under imperfect SIC decoding. Dinkelback method is first adopted to convert the objective of EE from the fractional form into a subtractive form. The converted problem is divided into subproblems and closed-form solutions are then derived based on dual method and Karush-Kuhn-Tucker (KKT) conditions. Simulation results show the benefit of our multi-cell NOMA BSC scheme compared to the benchmark multi-cell NOMA scheme in terms of system total EE. The main contributions of this paper are summarized as follow:

\begin{enumerate}
\item A new optimization framework for a multi-cell IoT network is considered, where a source in each cell transmits a superimposed signal to its serving IoT devices using NOMA protocol. A backscatter sensor tag (BST) in each cell also transmits data symbols towards nearby IoT devices by reflecting and modulating the superimposed signal of the source node. The objective is to maximize the total achievable EE of BSC network under im-SIC decoding. We simultaneously optimize the RC of BST, PAC of IoT devices, and total power budget of the source in each cell subject to the quality of services of IoT devices.
\item The optimization problem to maximize the total EE is formulated as a non-convex which is very complex and hard to be solved. Hence, the Dinkelbach method is first adopted to the original problem to convert the objective of EE from the fractional form into a subtractive form. The converted problem is then divided into two subproblems, i.e., power optimization at source and RC at BST in each cell. Next, we prove the RC subproblem as concave and exploit KKT conditions to obtain an efficient solution. Similarly, we prove that the power allocation subproblem is concave and solve it using the Lagrangian dual method. 
\item We also investigate the same model without BSC (also known as pure NOMA IoT network without backscattering) and set it as the benchmark framework. The numerical results for the proposed framework are corroborated by using Monte Carlo simulation which demonstrates the advantage of NOMA BSC network over the conventional NOMA IoT network without BSC. In addition, the proposed algorithm is less complex and converges after few iterations.
\end{enumerate}
The remaining of this paper is structured as follows: The system model and problem formulation are provided in Section II. The EE maximization solution is proposed in Section III. Simulation results and discussion is presented in Section IV followed by the concluding remarks in Section V. 
\begin{figure*}[!t]
\centering
\includegraphics [width=0.60\textwidth]{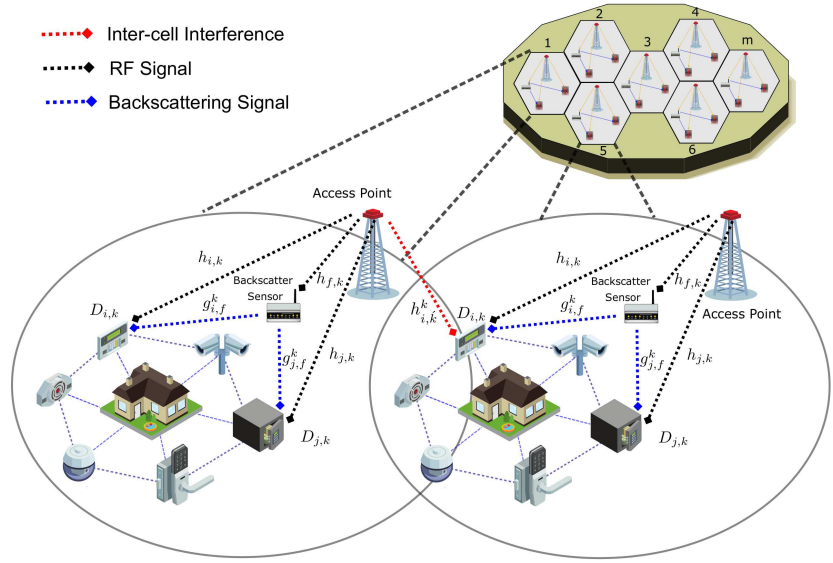}
\caption{Illustration of system model.}
\label{blocky}
\end{figure*}
\section{System Model and Problem Formulation}
We consider a multi-cell BSC network as shown in Fig. \ref{blocky}, wherein each cell, a source (denoted as $\mathcal S$) communicates with two downlink IoT devices ($\mathcal D_{i,k}$ and $\mathcal D_{j,k}$) using NOMA protocol\footnote{This work considers two IoT devices in each cell, however, it can be easily extended to the multi-user scenario. For instance, if the region of each cell is partitioned into multiple clusters and each cluster consists of two IoT devices. In such a case, NOMA is among IoT devices in the same cluster and OMA can be utilized between different clusters \cite{9219112}.}. The network also consists of $F$ uplink BSTs where it set can be denoted as $f=\{1,2,3,\dots F\}$. A BST in each cell also receives the downlink superimposed signal from $\mathcal S$, uses it to modulate information, and then reflects it towards IoT devices in the uplink direction, where the IoT devices also act as readers. The set of cells can be denoted as $K$ such as $k=\{1,2,3,\dots K\}$, where $k$ represents source $\mathcal S_k$. We assume that: 1) All the transmitters and receivers are using single antenna for communication; 2) all the sources reuse the same spectrum/time resources; 3) the channel state information of IoT devices in each cell is available at the source \cite{9261140}; 4) a decoding error can occur during the SIC process at receiver side such that the interference of other devices cannot be removed. Therefore, we consider SIC with decoding error. A superimposed signal $x_k$ transmitted by source $\mathcal S_k$ to $\mathcal D_{i,k}$ and $\mathcal D_{j,k}$ can be expressed as: 
\begin{align}
x_k=\sqrt{P_k\varLambda_{i,k}}x_{i,k}+\sqrt{P_k\varLambda_{j,k}}x_{j,k}, \label{1}
\end{align}
where $P_k$ is the transmit power of $\mathcal S_k$, $\varLambda_{i,k}$ and $\varLambda_{j,k}$ denote the PAC of $\mathcal S_k$. $x_{i,k}$ and $x_{j,k}$ are the unit power data symbols of $\mathcal D_{i,k}$ and $\mathcal D_{j,k}$ from $\mathcal S_k$. 
Meanwhile, the BST denoted as $\mathcal B_{f,k}$ also receives $x_k$ from $\mathcal S_k$, reflect it towards $\mathcal D_{i,k}$ and $\mathcal D_{j,k}$ by adding data symbol $w(t)$ such that $\mathbb E[|w(t)|^2]=1$, where $\mathbb E[.]$ represents the expectation operation. Therefore, $\mathcal D_{i,k}$ and $\mathcal D_{j,k}$ receive signals from both $\mathcal S_k$ and $\mathcal B_{f,k}$. Following the work in \cite{8540884}, If the channel from $\mathcal S_k$ to $\mathcal D_{i,k}$ and $D_{j,k}$ is modeled as $h_{i,k}=\bar{h}_{i,k}d^{-\varrho/2}_{i,k}$ and $h_{j,k}=\bar{h}_{j,k}d^{-\varrho/2}_{j,k}$, where $\bar{h}_{\varsigma,k}\sim\mathcal{CN}(0,1)$, $\varsigma\in\{i,j\}$ are the coefficient of Rayleigh fading, $d_{\varsigma,k}$ is the distance from $\mathcal S_k$ to $\mathcal D{i,k}$ and $\mathcal D_{\varsigma,k}$ and $\varrho$ shoes the path loss exponent. Then, the received signal of $\mathcal D_{i,k}$ and $\mathcal D_{j,k}$ can be written as:
\begin{align}
y_{i,k}&=\sqrt{h_{i,k}}x_k+\sqrt{\varPhi_{f,k}g^k_{i,f}(h_{f,k}}x_k)w(t)\nonumber\\&+\sum\limits_{k'=1, k'\neq k}^K\sqrt{P_{k'}h^k_{i,k'}}x_{k'}+\varpi_{i,k},\label{2}\\
y_{j,k}&=\sqrt{h_{j,k}}x_k+\sqrt{\varPhi_{f,k}g^k_{j,f}(h^k_{j,f}}x_k)w(t)\nonumber\\&+\sum\limits_{k'=1, k'\neq k}^K\sqrt{P_{k'}h^k_{j,k'}}x_{k'}+\varpi_{j,k},\label{3}
\end{align}
where in both (\ref{2}) and (\ref{3}), the first segment refer to the desired signal of $\mathcal S_k$, the second segment is the reflected signal of $\mathcal B_{f,k}$ and the third segment represents the inter-cell interference of neighboring cells. Further, $h_{f,k}$ is the channel gain between $\mathcal B_{f,k}$ and $\mathcal S_k$, $\varPhi_{f,k}$ refers to the RC of $\mathcal B_{f,k}$. Further, $g^k_{i,f}$ and $g^k_{j,f}$ denote the channel gains from $\mathcal B_{f,k}$ to $\mathcal D_{i,k}$ and $\mathcal D_{j,k}$. In addition, $P_{k'}$ is the interference power from $\mathcal S_{k'}$, $h^k_{i,k'}$ and $h^k_{j,k'}$ are the channel gains from $\mathcal S_{k'}$ to $\mathcal D_{i,k}$ and $\mathcal D_{j,k}$. Moreover, $\varpi_{i,k}$ and $\varpi_{j,k}$ are the additive white Gaussian noises (AWGN) with zero mean and $\sigma^2$ variance. According to the NOMA, $\mathcal D_{i,k}$ can decodes the signals $x_{i,k}$ and $w(t)$ by applying the SIC technique. In contrary, $\mathcal D_{j,k}$ cannot apply SIC and decodes the signal $x_{j,k}$ with interference. 

By considering the detecting and decoding sensitivity of receiver, $\mathcal{BS}_{k}$ a decoding error can occur during the SIC process at $\mathcal{D}_{i,k}$ such that the interference of $\mathcal{D}_{j,k}$ cannot be removed. Therefore, the received signal to interference plus noise ratio (SINR) of $\mathcal D_{i,k}$ when subtracting the signal of $\mathcal D_{j,k}$ can be given as:
\begin{align}
\gamma^k_{i\rightarrow j}=\frac{P_k\varLambda_{j,k}|h_{i,k}|^2+\varPhi_{f,k}|h_{f,k}|^2|g^k_{i,f}|^2}{P_k\varLambda_{i,k}(|h_{i,k}|^2+\varPhi_{f,k}|h_{f,k}|^2|g^k_{i,f}|^2)+\varDelta^k_{j,k'}+\sigma^2}, \label{4}
\end{align}
where $\varDelta^k_{j,k'}=|h^k_{j,k'}|^2\sum_{k'=1}^{K}P_{k'}$
is the inter-cell interference due to the co-channel deployment. The SINR at $\mathcal D_{i,k}$ to decode its own signal can be stated as:
\begin{align}
\gamma^k_{i\rightarrow i}=\frac{P_k\varLambda_{i,k}(|h_{i,k}|^2+\varPhi_{f,k}G_{i,k})}{P_k\varLambda_{j,k}|h_{i,k}|^2\beta+\varDelta^k_{i,k'}+\sigma^2}, \label{5}
\end{align}
where ${G_{i,k} = |h_{f,k}|^2|g^k_{i,f}|^2 }$. $\beta$ represents the imperfect SIC parameter which is given as $\beta=\mathbb E[|x_{i,k}-\tilde{x}_{i,k}|^2]$, where $x_{i,k}-\tilde{x}_{i,k}$ stands for the difference between the original and the estimated signals. The corresponding rate of $D_{i,k}$ can be written as
$R_{i,k}=\log_2(1+\gamma^k_{i\rightarrow i})$. The SINR at $\mathcal D_{j,k}$ to decode $x_{j,k}$ can be written as:
\begin{align}
\gamma^k_{j\rightarrow j}=\frac{P_k\varLambda_{j,k}(|h_{j,k}|^2+\varPhi_{f,k}G_{j,k})}{P_k\varLambda_{i,k}(|h_{j,k}|^2+\varPhi_{f,k}G_{j,k})+\varDelta^k_{j,k'}+\sigma^2}, \label{7}
\end{align}
where ${G_{j,k} = |h_{f,k}|^2|g^k_{j,f}|^2 }$. Thus, its corresponding data rate is can be written as
$R_{j,k}=\log_2(1+\gamma^k_{j\rightarrow j})$.

The objective of this work is to maximize the total EE of multi-cell NOMA BSC network. The total EE is given by
\begin{align}
EE=\sum\limits^K_{k=1}\left(\frac{R_k}{P_k\varLambda_{i,k}+P_k\varLambda_{j,k}+p_c}\right), \label{9}
\end{align}
where $R_k=R_{i,k}+R_{j,k}$ is the sum rate of $\mathcal S_k$ while the circuit power is represented by $p_c$. The EE of the system can be maximized through the efficient allocation of transmit power of $\mathcal S_k$, PAC of IoT devices, and the RC of BST in each cell. In addition, we also aim to ensure the minimum data rate of IoT devices in each cell. Mathematically, a joint optimization problem (P) is to maximize the total
EE of multi-cell NOMA BSC network can be formulated as:
\begin{alignat}{2}
&\text{(P)}\ \quad \underset{{(\varLambda_{i,k},\varLambda_{j,k},\varPhi_{f,k})}}{\text{max}} EE \label{10}\\
s.t.\quad&\text{C1}: \ P_k\varLambda_{i,k}\left(|h_{i,k}|^2+\varPhi_{f,k}G_{i,k}\right)\geq\left(2^{R_{min}}-1\right)\nonumber\\ & \times\left(|h_{i,k}|^2P_k\varLambda_{j,k}\beta+\varDelta^k_{i,k'}+\sigma^2\right),\ \forall k, \nonumber\\
&\text{C2}: \ P_k\varLambda_{j,k}\left(|h_{j,k}|^2+\varPhi_{f,k}G_{j,k}\right)\geq\left(2^{R_{min}}-1\right)\nonumber\\ & \times\left(P_k\varLambda_{i,k}(|h_{j,k}|^2+\varPhi_{f,k}|G_{j,k}\right)+\varDelta^k_{j,k'}+\sigma^2), \forall k, \nonumber\\
&\text{C3:} \ P_k\varLambda_{i,k}\leq P_k\varLambda_{j,k}, \ \forall k,\ \forall i,j,\nonumber\\
&\text{C4:} \ 0\leq P_k\leq P_{max},\ \forall k ,\nonumber\\
&\text{C5:} \ \varLambda_{i,k}+\varLambda_{j,k}\leq1,\ \forall k,\nonumber\\
&\text{C6:} \ 0\leq\varPhi_{f,k}\leq1,\ \forall f,\ \forall k,\nonumber
\end{alignat}
where constraints C1 and C2 guarantee the minimum data rate of IoT $\mathcal D_{i,k}$ and $\mathcal D_{j,k}$ associated with $\mathcal S_k$. Constraint C3 ensures the SIC decoding at receivers. Constraint C4 limits the transmit power of $\mathcal S_k$. Constraint C5 describes the condition for PAC of IoT devices connected to $\mathcal S_k$ while constraint C6 limits the RC of BST between 0 and 1.

\section{Energy Efficiency Maximization Solution}
The above EE maximization problem defined in (\ref{10}) is coupled on two variables in each cell, i.e., 1) Transmit power of the source and PAC of IoT devices in each cell, and 2) RC of BST in each cell. Thus, it is very hard to solve it directly. Therefore, this problem can be solved in three steps: i) First, we apply Dinkelbach method to transform the objective function of (P) into subtractive one; ii) second, on the fixed value of source transmit power in each cell, we compute the efficient RC of BST in each cell, and iii) third, we substitute the RC of BST in (8) and calculate the transmit power of source and PAC of IoT devices. Based on Dinkelbach method, the problem in (\ref{10}) can be transformed as: 
\begin{alignat}{2}
& \underset{{(\varLambda_{i,k},\varLambda_{j,k},\varPhi_{f,k})}}{\text{max}} \sum\limits_{k=1}^K{R_k}-\varPi\sum\limits_{k=1}^KP_k(\varLambda_{i,k}+\varLambda_{j,k})+p_c,\label{11}\nonumber \\
&s.t.\quad \text{C1}-\text{C6},
\end{alignat}
where $\varPi$ shows the maximum EE and it can be achieved when
\begin{align}
\sum\limits_{k=1}^K R_k -\varPi^*(\sum\limits_{k=1}^KP_k(\varLambda^*_{i,k}+\varLambda^*_{j,k})+p_c)=0.\label{12}
\end{align}
The problem in (\ref{11}) is still hard to be solved due to the interference terms in the SINR of $\mathcal D_{i,k}$ and $\mathcal D_{j,k}$ and the coupled variables $\varLambda_k$ and $\varPhi_{f,k}$. Thus, we decouple problem (\ref{11}) into two subproblems, i.e., RC selection subproblem and transmit power allocation subproblem.
\subsection{Efficient Reflection Coefficient Selection}
Here we compute the efficient RC of BST in each cell. For any given power allocation $\varLambda^*_k$ at $\mathcal S_k$ in each cell, the optimization problem in (\ref{11}) can be simplified to BST RC selection subproblem as: 
\begin{alignat}{2}
& \underset{{(\varPhi_{f,k})}}{\text{max}} \sum\limits_{k=1}^K\log_2\bigg\{\bigg(1+\frac{X_{i,k}+\varPhi_{f,k}Y_{i,k}}{Z_{i,k}}\bigg)\nonumber\\
& +\log_2\bigg(1+\frac{X_{j,k}+\varPhi_{f,k}Y_{j,k}}{Z_{j,k}+\varPhi_{f,k}W_{j,k}}\bigg)\bigg\} \nonumber\\
&-\varPi\sum\limits_{k=1}^KP_k(\varLambda^*_{i,k}+\varLambda^*_{j,k})+p_c,\label{13}\\
s.t.& \quad\text{C1}, \text{C2}, \text{C4}, \text{C6},\nonumber 
\end{alignat}
where $X_{i,k}=P_k\varLambda^*_{i,k}|h_{i,k}|^2$, $Y_{i,k}=P_k\varLambda^*_{i,k}G_{i,k}$, $Z_{i,k}={P_k\varLambda^*_{j,k}|h_{i,k}|^2\beta+\varDelta^k_{i,k'}+\sigma^2}$, $X_{j,k}=P_k\varLambda^*_{j,k}|h_{j,k}|^2$, $Y_{j,k}=P_k\varLambda^*_{j,k}G_{j,k}$, $Z_{j,k}={P_k\varLambda^*_{i,k}|h_{j,k}|^2\beta+\varDelta^k_{j,k'}+\sigma^2}$ and $W_{j,k}=P_k\varLambda^*_{i,k}G_{j,k}$. By using the following proposition, we demonstrate that $R_k$ is a concave/convex using $\varPhi_{f,k}$.
\begin{prop} The sum rate of $\mathcal S_k$
\begin{align}
&R_k=\log_2\bigg\{\bigg(1+\frac{X_{i,k}+\varPhi_{f,k}Y_{i,k}}{Z_{i,k}}\bigg)\nonumber\\
&+\log_2\bigg(1+\frac{X_{j,k}+\varPhi_{f,k}Y_{j,k}}{Z_{j,k}+\varPhi_{f,k}W_{j,k}}\bigg)\bigg\},
\end{align}
is concave/convex with reference to $\varPhi_{f,k}$.
\end{prop}
\begin{proof}
Refer to Appendix A.
\end{proof}
According to the Proposition 1, the optimization problem (\ref{13}) is concave which motivates us to exploit KKT conditions for obtaining optimal $\varPhi_{f,k}$. 
\begin{prop}
The closed-form of BST RC can be then expressed as:
\begin{align}
\varPhi_{f,k}= \left[\frac {(2^{\gamma^{min}_{i,k}}-1)-X_{i,k}}{Y_{i,k}}\right], \label{15}
\end{align}
\end{prop}
\begin{proof}
Please, refer to Appendix B.
\end{proof}
In the sequel, we calculate efficient power allocation in each cell.

\subsection{Efficient Power Allocation}
Here we calculate the efficient transmit power of source and PAC of IoT devices in each cell. For the fixed value of BST RC $\varPhi^*_{f,k}$, the optimization problem in (\ref{11}) can be then written as:
\begin{alignat}{2}
	& \underset{{(\varLambda_{i,k},\varLambda_{j,k})}}{\text{max}} \sum\limits_{k=1}^K {E_k}= \underset{{(\varLambda_{i,k},\varLambda_{j,k})}}{\text{max}} \sum\limits_{k=1}^K \frac{R_{i,k}+R_{j,k}}{P_k(\varLambda_{i,k}+\varLambda_{j,k})+p_c}\label{16}\nonumber \\
	& s.t.\quad \text{C1}-\text{C5}.
\end{alignat}
We can also write Equation (\ref{5}) and (\ref{7}) as:
\begin{align}
\gamma^k_{i\rightarrow i}=\frac{P_k\varLambda_{i,k}A_{i,k}}{P_k\varLambda_{j,k}{B_{i,k}}+C_{i,k}}, \label{17}
\end{align}
with
${A_{i,k}}=|h_{i,k}|^2+\varPhi_{f,k}G_{i,k}$ , ${B_{i,k}}=|h_{i,k}|^2\beta$,
${C_{i,k}}=\varDelta^k_{i,k'}+\sigma^2$,
and
\begin{align}
\gamma^k_{j\rightarrow j}=\frac{P_k\varLambda_{j,k}A_{j,k}}{P_k\varLambda_{i,k}B_{j,k}+C_{j,k}}, \label{18}
\end{align}
where
${A_{j,k}}=|h_{j,k}|^2+\varPhi_{f,k}G_{j,k}$ ,
 ${B_{j,k}}=|h_{j,k}|^2+\varPhi_{f,k}G_{j,k}$,
${C_{j,k}}=\varDelta^k_{j,k'}+\sigma^2$.
In the following proposition, we will prove that (\ref{16}) is concave/convex regarding $\boldsymbol{\varLambda_k}=\{\varLambda_{i,k}, \varLambda_{j,k}\}$.
\begin{prop} The sum rate of $\mathcal S_k$
\begin{align}
&R_k=\log_2(1+\frac{P_k\varLambda_{i,k}A_{i,k}}{P_k\varLambda_{j,k}{B_{i,k}}+C_{i,k}})\nonumber\\
&+\log_2(1+\frac{P_k\varLambda_{j,k}A_{j,k}}{P_k\varLambda_{i,k}B_{j,k}+C_{j,k}}) 
\end{align}
is concave/convex with reference to $\boldsymbol{\varLambda_k}=\{\varLambda_{i,k}, \varLambda_{j,k}\}$.
\end{prop}
\begin{proof}
The proof is demonstrate in Appendix C.
\end{proof}
Based on Proposition 3, the objective function in (\ref{16}) is concave-convex fractional programming problem, which can be solved through Dinkelbach algorithm as follow:
\begin{alignat}{2}
& \underset{{(\varLambda_{i,k},\varLambda_{j,k})}}{\text{max}} \sum\limits_{k=1}^K {E_k}= \underset{{(\varLambda_{i,k},\varLambda_{j,k})}}{\text{max}} \sum\limits_{k=1}^K F(\varPi) \nonumber\\
&= \underset{{(\varLambda_{i,k},\varLambda_{j,k})}}{\text{max}} \sum\limits_{k=1}^K R_k - \varPi(\sum\limits_{k=1}^KP_k(\varLambda_{i,k}+\varLambda_{j,k})-p_c)\label{19}\\
& s.t.\quad \text{C1}-\text{C5}.\nonumber 
\end{alignat}
where $\varPi=\frac{R_k}{\sum\limits_{k=1}^KP_k(\varLambda_{i,k}+\varLambda_{j,k})+p_c}$, while $F(\varPi)$ is the parametric form of fractional objective function in (\ref{19}). Solving the roots of $F(\varPi)$ is equivalent to computing the fractional objective function in (\ref{19}). $F(\varPi)$ as function of $\varPi$ is convex because it is negative when $\varPi$ tends to infinity and is positive when $\varPi$ approaches minus infinity. Therefore, motivated by the above observations, this convex problem can be solved through Lagrangian dual decomposition method. The Lagrangian function of problem (\ref{19}) can be defined as:
\begin{alignat}{2}
& L(\boldsymbol{\varLambda_k},\boldsymbol{\lambda_k},\mu_k,\epsilon_k)=\sum\limits_{k=1}^K\bigg\{\log_2\bigg(1+\frac{P_k\varLambda_{i,k}A_{i,k}}{P_k\varLambda_{j,k}{B_{i,k}}+C_{i,k}}\bigg)\nonumber\\
&+\log_2\bigg(1+\frac{P_k\varLambda_{j,k}A_{j,k}}{P_k\varLambda_{i,k}B_{j,k}+C_{j,k}}\bigg)\bigg\}\nonumber\\
&-\varPi\sum\limits_{k=1}^KP_k(\varLambda_{i,k}+\varLambda_{j,k})-p_c+\lambda_{i,k}(P_k\varLambda_{i,k}A_{i,k}\label{20}\\
&-(2^{R_{min}}-1){P_k\varLambda_{j,k}{B_{i,k}}+C_{i,k}})+\lambda_{j,k}(P_k\varLambda_{j,k}A_{j,k}\nonumber\\
&-(2^{R_{min}}-1)(P_k\varLambda_{i,k}B_{j,k}+C_{j,k})+\mu_k(P_{max}-P_k)\nonumber\\
&+\epsilon_k(1-\varLambda_{i,k}-\varLambda_{j,k}),\nonumber
\end{alignat}
where $\boldsymbol{\lambda_k}=\{\lambda_{i,k}, \lambda_{j,k}\}$, $\mu_k$, and $\epsilon_k$ are the dual variables, which are related to the constraints C1, C2, C4, and C5. The Lagrangian dual function can be presented as:
\begin{align}
g(\boldsymbol{\lambda_k},\mu_k, \epsilon_k) = \underset{{\boldsymbol{\varLambda_k}>0,\boldsymbol{\lambda_k},\mu_k,\epsilon_k\ge 0}}{\text{max}}L(\boldsymbol{\varLambda_k},\boldsymbol{\lambda_k},\mu_k,\epsilon_k) \label{21} 
\end{align}
Then, its a dual problem can be formulated as follow:
\begin{align}
\underset{{\boldsymbol{\lambda_k},\mu_k,\epsilon_k\ge 0}}{\text{min}}g(\boldsymbol{\lambda_k},\mu_k, \epsilon_k) \label{22} 
\end{align}
For the fixed dual variables and given EE $\varPi$, the formulated optimization problem depends on KKT conditions. 
\begin{prop}
The closed-form expression for energy-efficient PAC of $\mathcal D_{i,k}$ and $\mathcal D_{j,k}$ can be derived as:
\begin{align}
&\varLambda^*_{i,k}= \left[\frac {-b \pm \sqrt{b^2-4ac}}{2a}\right]^+ \label{23}\\ 
&\varLambda^*_{j,k}= 1- \varLambda^*_{i,k} \label{24}
\end{align}
where $[.]^+=\text{max}[0,.]$ and the values of 
\begin{align}
& a=P_{k}^2 (-A_{i,k} A_{j,k} B_{j,k} (1 + \lambda_{i,k}) (C_{i,k} + B_{i,k} P_{k})+ A_{i,k}B_{j,k}^2 \nonumber\\& (1 + \lambda_{i,k}) (C_{i,k} + B_{i,k} P_{k}) +  A_{i,k} A_{j,k} B_{i,k}(1 +\lambda_{j,k})(C_{j,k} \nonumber\\&  + B_{j,k} P_{k}) -  A_{j,k} B_{i,k}^2 (1 + \lambda_{j,k}) (C_{j,k} + B_{j,k}P_{k})),\\
& b=P_{k} (C_{i,k} + B_{i,k} P_{k}) (-A_{i,k} C_{j,k} (-2 B_{j,k} (1 + \lambda_{i,k}) \nonumber\\&+ A_{j,k}(2 + L_{i,k} + \lambda_{j,k})) + A_{i,k} A_{j,k} B_{j,k} (\lambda_{i,k} - \lambda_{j,k}) P_{k} \nonumber\\&+ 2A_{j,k} B_{i,k} (1 + \lambda_{j,k}) (C_{j,k} + B_{j,k} P_{k})),\\
& c=(C_{i,k} + B_{i,k} P_{k}) (A_{i,k} C_{j,k}^2 (1 + \lambda_{i,k})+ A_{j,k} \nonumber\\& (-C_{i,k} (1 +\lambda_{j,k}) (C_{j,k} + B_{j,k} P_{k})+ P_{k} (A_{i,k} C_{j,k}\nonumber\\&  (1 + \lambda_{i,k}) - B_{i,k}(1 + \lambda_{i,k}) (C_{j,k} + B_{j,k} P_{k})))).
\end{align}
\end{prop}
\begin{proof}
Please, refer to Appendix D.
\end{proof}
Next we calculate the optimal transmit power of each source, i.e., $P_k$. To do so, we differentiate (\ref{20}) with respect to $P_{k}$, it results as:
\begin{align}
\tau+\chi P_{k}+\psi P_{k}^2+\Gamma P_{k}^3+\omega P_{k}^4=0,\label{neweq}
\end{align}
where $\tau, \chi, \psi, \Gamma$ and $\omega$ are given in (26)-(29) on the top of the next page.
\begin{figure*}
\begin{align}
\tau &= C_{i,k} C_{j,k}(-A_{j,k} C_{i,k} (-1 + \Lambda_{i,k})) (1+\lambda_{j,k})+C_{j,k}(A_{i,k} \Lambda_{i,k} (1+\lambda_{i,k})-C_{i,k} (\mu_{k}+\varPi)),\\
\chi&= C_{i,k} C_{j,k} (2 (B_{i,k} C_{j,k} (-1+\Lambda_{i,k}) -B_{j,k} C_{i,k} \Lambda_{i,k}) (\mu_{k}+\varPi)+A_{j,k} (-1 + \Lambda_{i,k}) (2 B_{i,k} (-1 + \Lambda_{i,k}) (1 + \lambda_{j,k})\nonumber\\& - A_{i,k} \Lambda_{i,k} (2 + \lambda_{i,k} + \lambda_{j,k}) + C_{i,k} (\mu_{k} + \varPi)) +   A_{i,k}\Lambda_{i,k} (2 B_{j,k} \Lambda_{i,k} (1 + \lambda_{i,k}) - C_{j,k} (\mu_{k} + \varPi))),\\
\psi&=-(B_{i,k}^2 C_{j,k}^2 (-1 + \Lambda_{i,k})^2 - 4 B_{i,k} B_{j,k} C_{i,k} C_{j,k} (-1 +\Lambda_{i,k}) \Lambda_{i,k} + B_{j,k}^2 C_{i,k}^2 \Lambda_{i,k}^2) (\mu_{k} + \varPi) + A_{i,k} \Lambda_{i,k} (B_{j,k}^2\nonumber\\& C_{i,k} \Lambda_{i,k}^2 (1 + \lambda_{i,k}) +  B_{i,k} C_{j,k}^2 (-1 + \Lambda_{i,k}) (\mu_{k} + \varPi)-  2 B_{j,k} C_{i,k} C_{j,k} \Lambda_{i,k} (\mu_{k} + \varPi)) -  A_{j,k} (-1 + \Lambda_{i,k})\nonumber\\& (B_{i,k}^2 C_{j,k} (-1 + \Lambda_{i,k})^2 (1 + \lambda_{j,k}) -  B_{i,k} C_{j,k} (-1 +\Lambda_{i,k}) (A_{i,k} \Lambda_{i,k} (1 + \lambda_{j,k}) - 2 C_{i,k} (\mu_{k} + \varPi)) - C_{i,k} \nonumber\\&\Lambda_{i,k} (B_{j,k} C_{i,k} (\mu_{k} + \varPi) + A_{i,k} (-B_{j,k} \Lambda_{i,k} (1 + \lambda_{i,k}) + C_{j,k} (\mu_{k} + \varPi)))),\\      
\Gamma&=(B_{j,k} \Lambda_{i,k} (-2 B_{i,k}^2 C_{j,k} (-1 + \Lambda_{i,k})^2 + 2 B_{i,k} (B_{j,k} C_{i,k}+ A_{i,k} C_{j,k}) (-1 + \Lambda_{i,k}) \Lambda_{i,k} -A_{i,k} B_{j,k} C_{i,k} \Lambda_{i,k}^2) +A_{j,k}  \nonumber\\&(-1 + \Lambda_{i,k}) (B_{i,k}^2 C_{j,k} (-1 + \Lambda_{i,k})^2 - B_{i,k} (2 B_{j,k} C_{i,k} + A_{i,k} C_{j,k}) (-1 + \Lambda_{i,k}) \Lambda_{i,k} + A_{i,k} B_{j,k} C_{i,k} \Lambda_{i,k}^2)) (\mu_{k} + \varPi),\\
\omega&=-B_{i,k} B_{j,k} (-1 +\Lambda_{i,k}) \Lambda_{i,k} (B_{i,k} (-1 + \Lambda_{i,k}) - A_{i,k} \Lambda_{i,k}) (A_{j,k} - A_{j,k} \Lambda_{i,k} + B_{j,k} \Lambda_{i,k}) (\mu_{k} + \varPi).     
\end{align}\hrulefill
\end{figure*}
Equation (\ref{neweq}) is the polynomial of order four which can be easily solved by any conventional solver. The objective of the problem is to maximize the EE, thus, $P_{k}^{*}$ can be founded through the larger root of (\ref{neweq}). With optimal $\varLambda^*_{i,k}$, $\varLambda^*_{j,k}$ and $P_{k}^*$ , problem (\ref{21}) can be written as:
\begin{alignat}{2}
& \underset{{(\varLambda^*_{i,k},\varLambda^*_{j,k}, P_{k}^*)}}{\text{max}} \sum\limits_{k=1}^K\bigg\{\log_2\bigg(1+\frac{P_k^*\varLambda^*_{i,k}A_{i,k}}{P_k^*\varLambda^*_{j,k}{B_{i,k}}+C_{i,k}}\bigg)\nonumber\\
&+\log_2\bigg(1+\frac{P_k^*\varLambda^*_{j,k}A_{j,k}}{P_k^*\varLambda^*_{i,k}B_{j,k}+C_{j,k}}\bigg)\bigg\}\nonumber\\
&-\varPi\sum\limits_{k=1}^K P_k^*(\varLambda^*_{i,k}+\varLambda^*_{j,k})+p_c,\\
&\text{subject to:}\quad {\boldsymbol{\lambda_k},\epsilon_k\ge 0} \nonumber 
\end{alignat}
Subsequently, we use sub-gradient method to iteratively update the Lagrangian multipliers $\lambda_{i,k}$, $\lambda_{j,k}$, $\mu_k$ and $\epsilon_k$ as \cite{9134383}:
\begin{align}
\lambda_{i,k}(t+1)&=\lambda_{i,k}(t)+\delta(t)(P_k^*\varLambda^*_{i,k}A_{i,k}\nonumber\\
&-(2^{R_{min}}-1){P_k^*\varLambda^*_{j,k}{B_{i,k}}+C_{i,k}}), \forall k, \label{sub1}
\end{align}
\begin{align}
\lambda_{j,k}(t+1)&=\lambda_{j,k}(t)+\delta(t)(P_k^*\varLambda^*_{j,k}A_{j,k}\nonumber\\
&-(2^{R_{min}}-1)(P_k^*\varLambda^*_{i,k}B_{j,k}+C_{j,k}),\forall k,\label{sub2}
\end{align}
\begin{align}
\epsilon_{k}(t+1)=&\epsilon_{k}(t)+\delta(t)(1-(\varLambda^*_{i,k}+\varLambda^*_{j,k})),\forall k,\label{sub3} 
\end{align}
\begin{align}
\mu_{k}(t+1)=&\mu_{k}(t)+\delta(t)(P_{max}-P_{k}^*),\forall k,\label{sub4new} 
\end{align}
where $t$ is the index of iteration. Equations (\ref{sub1}), (\ref{sub2}), (\ref{sub3}), and (\ref{sub4new}) are iteratively calculated until the required criterion satisfied.

\subsection{Proposed Algorithm and Complexity Analyses}
Here, we design algorithm based on the solutions provided in Section III-A and B, respectively. As shown in Algorithm 1, we initialize all the system parameters and variables. For the given values of $P_k$, $\varLambda_{i,k}$ and $\varLambda_{j,k}$, we compute $\varPhi_{f,k}$. Subsequently, we substitute the value of $\varPhi^*_{f,k}$ in power allocation subproblem (14) and calculate $\varLambda_{i,k}$ and $\varLambda_{j,k}$ followed by $P_k$. Then, we iteratively update $\lambda_{i,k}$, $\lambda_{j,k}$ and $\varepsilon_k$. The above process will continue until convergence criteria satisfied.

The computational complexity of the proposed optimization framework can be calculated regarding the number of iterations. The complexity of our algorithm depends on the different variables and parameters of the system such number of cells and the number of users, i.e., $K,I,J$. Based on these observation, the complexity of the proposed algorithm in any given iteration is computed as $\mathcal O[(I+J)K]$. Since we have considered two user in each cell, therefore, the computational complexity can also be written as $\mathcal O[2K]$. Further, if the number of total iteration required for convergence is $T$, the total computational complexity becomes $\mathcal O[2TK]$.
   \begin{algorithm}[t]
   \eIf{$t=0$}{
        Initialize all parameters and variables, i.e., number of cells, number of users, number of BSTs, maximum power budget of each source, RC of each BST, variance, minimum data rate, circuit power, values of imperfect SIC and channel gains.
         }{First we calculate the reflection coefficient of backscatter tag in each cell for the given values of $\varLambda_{i,k}$,$\varLambda_{i,k}$ and $P_k$
         
    \For{$k=1:K$}{Find $\varPhi_{f,k}$ according to (13) 
    }Next we substitute the value of $\varPhi^*_{f,k}$ in (14) and calculate the values of $\varLambda_{i,k}$,$\varLambda_{i,k}$ and $P_k$
    
    \While{not converge}{\For{$k=1:K$, $i=1:I$, $j=1:J$}{Compute $\varLambda_{i,k}$ according to (22) and $\varLambda_{j,k}$ according to (23) \\
    Compute $P_k$ according to (27)\\
    Update the dual variables $\lambda_{i,k}$, $\lambda_{j,k}$ and $\epsilon_{k}$
        }
    }Return $P^*_k$, $\varLambda^*_{i,k}$, $\varLambda^*_{j,k}$, $\varPhi_{f,k}$}
    \caption{Proposed resource optimization algorithm.}
   \end{algorithm}  
\section{Numerical Results and Discussion}

In this section, we provide the simulation results to evaluate the performance of the proposed framework. Unless specified otherwise the system parameters are taken as follow: $\sigma$=0.01,  $\beta$=0.1, K=10, $p_{c}$=0.1 W, and $P_{max}$= 32 dBm. In the considered problem, the power allocation is always lower bounded by the $R_{min}$ in order to make the impact of changing $P_{max}$ more prominent. In the first three results of this section, we have taken $R_{min}$=0. To analyze the benefits of backscattering, the performance of the proposed framework WBS (with backscatter sensor) is compared with a simplified network with no backscatter sensor (NBS)\footnote{Due to the novelty of the proposed framework, it is difficult to compare it with the existing works of the literature. Thus, we resort to compare it with pure NOMA without backscatter communication.}. 
\begin{figure}[!t]
\centering
\includegraphics [width=0.48\textwidth]{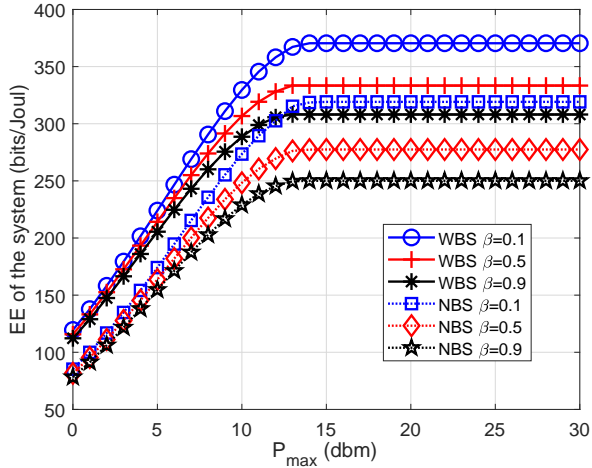}
\caption{The impact of increasing $P_{max}$ on the total EE of the system with different values of $\beta$}
\label{f2}
\end{figure}
\begin{figure}[!t]
\centering
\includegraphics [width=0.48\textwidth]{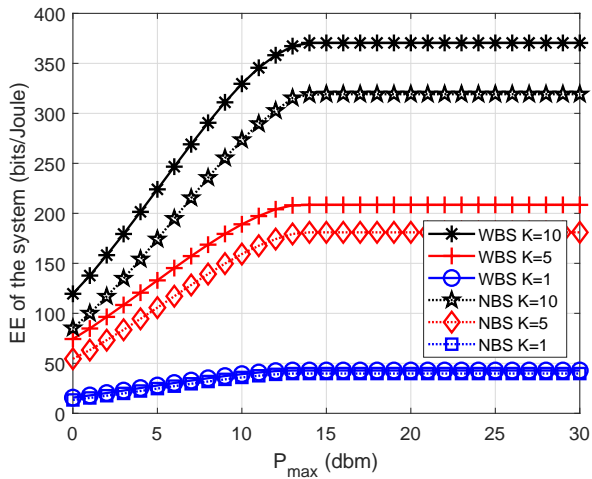}
\caption{The effect of increasing $P_{max}$ on the total EE for different number of cell in the system}
\label{f3}
\end{figure}

The effect of increasing $P_{max}$ on the total EE of the system is presented in Fig. \ref{f2}. An increase in the value of $P_{max}$ results in increasing the EE of the system initially. However, after a certain point, an increase in $P_{max}$ has no impact on the total EE. This is because, at these points, the transmission power is efficient, and allocating more power for the transmission results in decreasing the total EE of the system. Thus, when the value of $P_{max}$ is further increased, the allocated power for the transmission remains unchanged. Further, it can be seen that smaller values of $\beta$ result in providing more EE. The reason is that, at small $\beta$, less interference is faced by the near IoT devices, whereas, increasing $\beta$ would increase the SIC error resulting in the reduction of the overall system EE. At smaller values of $P_{max}$, the transmission power is also very less, this cause a very small interference to the other IoT devices in the system. Thus, the values of total EE for different $\beta$ have a very small gap, for smaller values of $P_{max}$. However, this gap increases with the increasing $P_{max}$, i.e., as the transmission power increases the interference also increase and the impact of imperfect SIC on the EE becomes more prominent. It is clear from Fig. \ref{f2} that the system with BST outperforms the network with no BST for all values of $P_{max}$. 

The total EE of the system also depends on total cells in the network. The impact of increasing $P_{max}$ on the EE of system containing different $K$ is shown in the Fig. \ref{f3}. For any value of K, increasing the value of $P_{max}$ increases the EE initially, however the efficiency becomes constant after a certain point because the transmission power remains unchanged. It is interesting to see that the difference between the EE offered the WBS and NBS systems increases with increasing $P_{max}$. This is because when the transmission power increases, the interference faced by all the IoT devices also increase. In the case of WBS systems, the increased transmission power also results in increasing the BST rate. Hence, the increase in the EE of WBS is more, as compared to the NBS. With more number of $K$ in the system, the benefit of BST becomes more clear, as it clear from the Fig. \ref{f3} that the gap between WBS and NBS increases with
\begin{figure}[!t]
\centering
\includegraphics [width=0.48\textwidth]{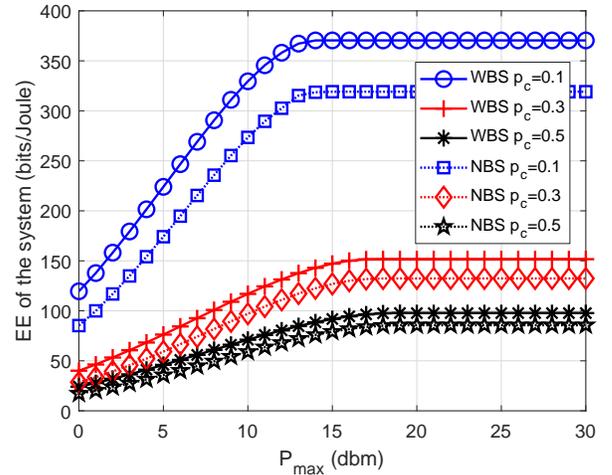}
\caption{The effect of increasing $P_{max}$ on the EE for different $p_c$}
\label{f5}
\end{figure}

The circuit power consumption ($p_c$) also affects the EE of the network. Fig. \ref{f5} shows that larger values of $p_c$ decreases the EE of the system. Further, it is clear from the figure that the WBS systems outperform the NBS for all values of $P_{max}$. An important point to note in Fig. \ref{f5} is that, when $p_c$ is increased the optimal value of power allocation is achieved at comparatively greater value of $P_{max}$. As in the case of $p_c$=0.1, the EE becomes constant at $P_{max}$= 16 dBm. However, for $p_c$=0.3 and $p_c$=0.5, the convergence behavior of EE is observed for $P_{max}\geq$ 19 and 21, respectively. This shows that for smaller values of $p_c$ consumption, the optimal behavior of the network is obtained with small values of $P_{max}$.
\begin{figure}[!t]
\centering
\includegraphics [width=0.48\textwidth]{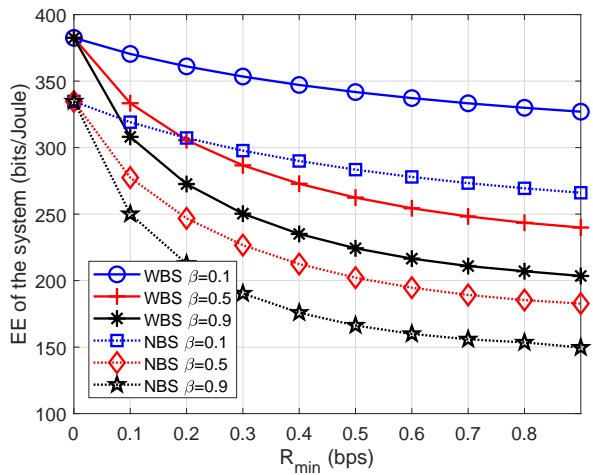}
\caption{The impact of increasing $R_{min}$ on the EE with different $\beta$}
\label{f6}
\end{figure}

The effect of increasing required rate of IoT devices ($R_{min}$) on the system EE shown in Fig. \ref{f6}. It is observed that the total EE decreases with the increasing values of $R_{min}$. The possible reason for this is the increase in the transmit power to satisfy the required rate of those IoT devices with weaker channel gains. However, this will reduce the overall EE of the network.
If the rate requirement can not be satisfied by varying $\Lambda_{i,k}$ then the system increases $P_{k}$ which results in further decreasing the EE. Data rate of IoT devices is a logarithmic function of power, hence, allocating more power to meet the rate requirement it results in decreasing the EE of the system. This can also be seen from the EE definition in (\ref{9}), as the numerator increases logarithmically with power and the increase in denominator is linear, when power allocation is increased beyond the optimal point, the EE of the system reduces. Another interesting thing to note in the Fig. \ref{f6} is that when $R_{min}$ is increased the EE of NBS decreases more rapidly as compared to WBS. This is because in NBS system when the allocated power is increased, it also results in increasing the interference which further adds to decrease the EE. However, in the case of WBS, allocating more power also enhances the data rate of BST rate, so this compensates for the increased interference to some extent. 

\begin{figure}[!t]
\centering
\includegraphics [width=0.48\textwidth]{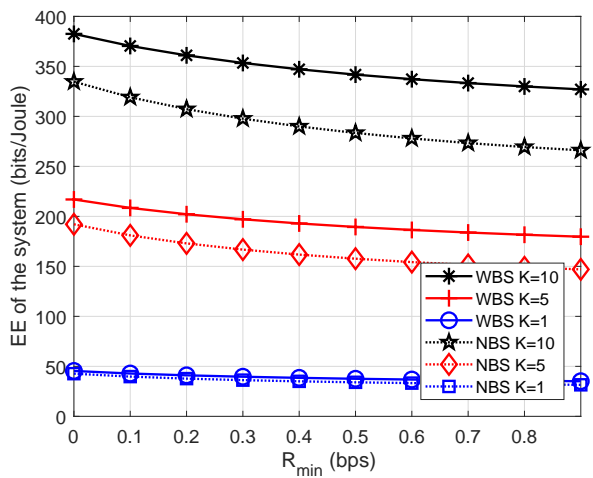}
\caption{The effect of increasing $R_{min}$ on the EE with different number of cells}
\label{f7}
\end{figure}
\begin{figure}[!t]
\centering
\includegraphics [width=0.48\textwidth]{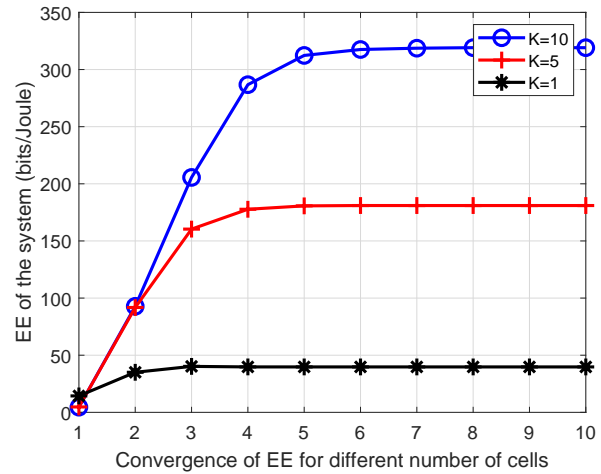}
\caption{Convergence of EE for different number of cells}
\label{f9}
\end{figure}
The results in Fig. \ref{f7} depicts the gap of system EE of WBS and NBS. The performance gap increases with the increase in K. This is because when the number of $K$ is increased, the total BST in the system also increases, so the benefit of BSC becomes more prominent. Another point to note is that in the case of WBS when $R_{min}$ is increased, the reduction in the case of K=10 is more rapid as compared to the system with 5 cells, and the network with only one cell faces the least reduction in the EE. This is because the IoT devices in the system having more cells receive more interference. Thus, when the allocated power is increased to satisfy the minimum rate requirement, the EE decreases rapidly (for K=10 each user faces interference from 9 other cells, whereas in the case of K=5, the inter-cell interference is caused by 4 cells). In the case of K=1 the users only face intra-cell interference which is somewhat compensated by the BST, hence in this case, the decrease in the EE is minimal.

The convergence behavior of the proposed framework is shown in Fig. \ref{f9}. When the number of the cells in the system are increased, the number of optimization variables increases which results in increasing the complexity of the system. The Fig. \ref{f9} shows that the system with just one cells takes the least number of iterations to converge, whereas the system with 10 cells is the slowest to converge. However, it can be seen that for any number of cells in the system, the proposed framework converges within limited iterations.

\section{Conclusion}
BSC and NOMA are the two emerging technologies to connect large-scale low-powered IoT devices in coming 6G era. In this paper, we has proposed the EE maximization approach for multi-cell NOMA BSC under the assumption of imperfect SIC. In particular, the power of source, PAC of IoT devices and RC of BST in each each have been jointly optimized to maximize total EE of the network. The Dinkelbach's algorithm has adopted first to transform the optimization followed by KKT conditions and dual method to obtain the efficient solutions. The simulation results has shown that the proposed multi-cell NOMA BSC framework outperforms the benchmark optimization framework and converges in a few iterations.


\section*{Appendix A: Proof of Proposition 1}
Here, concavity/convexity of $R_k$ w.r.t. $\varPhi_{f,k}$ is proved. The first derivative of $R_k$ w.r.t. $\varPhi_{f,k}$ is given as:
\begin{align}
& \frac {\partial R_k}{\partial \varPhi_{f,k}} = \frac {Y_{i,k}}{\ln(2)(A_{i,k}+Z_{i,k})}+\frac {C_{j,k}}{\ln(2)({B_{j,k}}^2+B_{j,k}A_{j,k})} \label{28}
\end{align}
where $A_{i,k}=(X_{i,k}+\varPhi_{f,k}Y_{i,k})$, $A_{j,k}=(X_{j,k}+\varPhi_{f,k}Y_{j,k})$, $B_{j,k}=(Z_{j,k}+\varPhi_{f,k}W_{j,k})$, and $C_{j,k}=({Y_{j,k}Z_{j,k}-X_{j,k}W_{j,k}})$. 
It's second order derivative is as: 
\begin{align}
& \frac {\partial^2R_k}{\partial {\varPhi_{f,k}}^2} =- \Biggl( \frac {{Y_{i,k}}^2}{\ln(2)(A_{i,k}+Z_{i,k})^2} \nonumber\\
&+\frac{C_{j,k}(2W_{j,k}E_{j,k}+C^+_{j,k})}{\ln(2)B_{j,k}^2({B_{j,k}}+B_{j,k}A_{j,k})^2}\Biggl) < 0 \label{29}
\end{align} 
where $E_{j,k}=B_{j,k}+Y_{j,k}\varPhi_{f,k}$, and $C^+_{j,k}=({Y_{j,k}Z_{j,k}+X_{j,k}W_{j,k}})$.\\
We can see that (39) is negative, therefore $R_k$ is concave/convex and is an increasing function with $\varPhi_{f,k}$.

\section*{Appendix B: Proof of Proposition 2}
We employ the dual method to obtain an efficient closed-form solution for convex optimization problem in (\ref{13}) with respect to RC of BST. The Lagrangian function of problem (\ref{13}) can be defined as:
\begin{alignat}{2}
& L(\varPhi_{f,k},\lambda_{i,k},\lambda_{j,k},\mu_k,\eta_{f,k})= \sum\limits_{k=1}^K N_k(\varPhi_{f,k})\nonumber\\ 
&-\varPi\sum\limits_{k=1}^KP_k(\varLambda_{i,k}+\varLambda_{j,k})-p_c+\lambda_{i,k}Q(\varPhi_{f,k},i,k)\nonumber\\
&+\lambda_{j,k}Q(\varPhi_{f,k},j,k)+\mu_k(P_{max}-P_k)+ \eta_{f,k}R(\varPhi_{f,k})\label{30}
\end{alignat}
where
\begin{alignat}{2}
&N_k(\varPhi_{f,k})= \log_2\bigg\{\bigg(1+\frac{X_{i,k}+\varPhi_{f,k}Y_{i,k}}{Z_{i,k}}\bigg)\nonumber\\
& +\log_2\bigg(1+\frac{X_{j,k}+\varPhi_{f,k}Y_{j,k}}{Z_{j,k}+\varPhi_{f,k}W_{j,k}}\bigg)\bigg\} \label{31}
\end{alignat}
\begin{alignat}{2}
&Q(\varPhi_{f,k},i,k)={X_{i,k}+\varPhi_{f,k}Y_{i,k}}-(2^{R_{min}}-1){Z_{i,k}} \label{32}
\end{alignat}
\begin{alignat}{2}
&Q(\varPhi_{f,k},j,k)={X_{j,k}+\varPhi_{f,k}Y_{j,k}}-(2^{R_{min}}-1)\nonumber\\ &\times({Z_{j,k}+\varPhi_{f,k}W_{j,k}}) \label{33}\\
& R(\varPhi_{f,k})={\varPhi_{f,k}-1} \label{34}
\end{alignat}
In (\ref{30}) $\lambda_{i,k}$, $\lambda_{j,k}$, $\mu_k$ and $\eta_{f,k}$ are called the Lagrangian multipliers. Next, we exploit the KKT conditions such as:
\begin{align}
\frac{\partial L(\varPhi_{f,k},\lambda_{i,k},\lambda_{j,k},\mu_k,\eta_{f,k})}{\partial \varPhi_{f,k}}|_{\varPhi=\varPhi^*}=0, \label{35}
\end{align}
The above equation results in
\begin{align}
& \frac {Y_{i,k}}{\ln(2)(A_{i,k}+Z_{i,k})}+\frac {C_{j,k}}{\ln(2)({B_{j,k}}^2+B_{j,k}A_{j,k})}+\lambda_{i,k}Y_{i,k}\nonumber\\
&+\lambda_{j,k}(Y_{j,k}-(2^{R_{min}}-1)W_{j,k})+\eta_{f,k} = 0 \label{36}
\end{align}
\begin{align}
& \frac {Y_{i,k}}{\ln(2)(A_{i,k}+Z_{i,k})}+\frac {C_{j,k}}{\ln(2)({B_{j,k}}^2+B_{j,k}A_{j,k})}+\eta_{f,k}\nonumber\\
&=(\lambda_{j,k}(2^{R_{min}}-1)W_{j,k} -\lambda_{j,k}Y_{j,k})-\lambda_{i,k}Y_{i,k} \label{37}
\end{align}
In (\ref{36}), $C_{j,k}= Y_{j,k}Z_{j,k}-X_{j,k}W_{j,k} = \varDelta^k_{j,k'}+\sigma^2 > 0$. The left hand side of (\ref{36}) is always positive and therefore
\begin{align}
(\lambda_{j,k}(2^{R_{min}}-1)W_{j,k} -\lambda_{j,k}Y_{j,k}) > \lambda_{i,k}Y_{i,k} \label{38}
\end{align}
In (\ref{38}), $(\lambda_{j,k}(2^{R_{min}}-1)W_{j,k} -\lambda_{j,k}Y_{j,k})$ is always positive because $(2^{R_{min}}-1)$ is always positive and $W_{j,k} > Y_{j,k}$. Since $\lambda_{i,k} \ge 0$, the $\lambda_{j,k}$ is nonnegative. The slack complimentary condition in KKT conditions is satisfied. Therefore, $Q(\varPhi_{f,k},i,k)$ and $Q(\varPhi_{f,k},j,k)$ corresponding to $\lambda_{i,k}$ and $\lambda_{j,k}$ are active. Hence, $Q(\varPhi_{f,k},i,k)=0$ and $Q(\varPhi_{f,k},j,k)=0$. Finally, the optimum $\varPhi_{f,k}$ is obtained from active inequality constraint as given in (\ref{15}).

\section*{Appendix C: Proof Of Proposition 3}
Here, the concavity/convexity of $\varLambda_{i,k}$ and $\varLambda_{j,k}$ is proved. The Hessian matrix should be negative definite, if a function is concave. The Hessian matrix is negative definite, when its principal minors have alternative signs. Here we derive a Hessian matrix for our formulated problem and demonstrate it as negative definite. The sum-rate of $\mathcal S_k$ can be written as:
\begin{align}
&R_k=\log_2(1+\gamma^k_{i\rightarrow i})+\log_2(1+\gamma^k_{j\rightarrow j}) \label{39}\\
R_k & =\log_2(1+\frac{P_k\varLambda_{i,k}A_{i,k}}{P_k\varLambda_{j,k}{B_{i,k}}+C_{i,k}})\nonumber\\
&+\log_2(1+\frac{P_k\varLambda_{j,k}A_{j,k}}{P_k\varLambda_{i,k}B_{j,k}+C_{j,k}}) \label{40}
\end{align}
The Hessian matrix of (\ref{40}) is defined as:
\begin{align}
H=
\begin{bmatrix}
\frac {\partial R_k}{\partial^2\varLambda_{i,k}} & \frac {\partial R_k}{\partial\varLambda_{i,k}\partial\varLambda_{j,k}} \\
\frac {\partial R_k}{\partial\varLambda_{j,k}\partial\varLambda_{i,k}} & \frac {\partial R_k}{\partial^2\varLambda_{i,j}} \label{41}
\end{bmatrix}
\end{align}
\begin{align}
&\frac {\partial R_k}{\partial^2\varLambda_{i,k}}\nonumber=\varphi_{1,1}= - \\
&\frac {A_{i,k}^2V_{j,k}^2T_{j,k}^2-A_{j,k}B_{j,k}^2T_{i,k}^2(2V_{j,k}+A_{j,k}\varLambda_{j,k})\varLambda_{j,k}}{\ln(2)T_{i,k}^2T_{j,k}^2V_{j,k}^2} \label{42}
\end{align}
\begin{align}
&\frac {\partial R_k}{\partial\varLambda_{i,k}\partial\varLambda_{j,k}}=\varphi_{1,2} \nonumber\\
&= -\frac {A_{i,k}B_{i,k}T_{j,k}^2-A_{j,k}B_{j,k}T_{i,k}^2}{\ln(2)T_{i,k}^2T_{j,k}^2} \label{43}\\
&\frac {\partial R_k}{\partial^2\varLambda_{j,k}} =\varphi_{2,1} \nonumber\\ 
&= -\frac {A_{j,k}^2V_{i,k}^2T_{i,k}^2-A_{i,k}B_{i,k}^2T_{j,k}^2(2V_{i,k}+A_{i,k}\varLambda_{i,k})\varLambda_{i,k}}{\ln(2)T_{j,k}^2T_{i,k}^2V_{i,k}^2} \label{44}\\
&\frac {\partial R_k}{\partial\varLambda_{j,k}\partial\varLambda_{i,k}}=\varphi_{2,2} \nonumber\\
&= -\frac {A_{i,k}B_{i,k}T_{j,k}^2-A_{j,k}B_{j,k}T_{i,k}^2}{\ln(2)T_{i,k}^2T_{j,k}^2} \label{45}
\end{align}
where, $T_{i,k}= A_{i,k}\varLambda_{i,k}+V_{i,k}$, $T_{j,k}= A_{j,k}\varLambda_{j,k}+V_{j,k}$, $V_{i,k}= B_{i,k}\varLambda_{j,k}+C_{i,k}$, and $V_{j,k}= B_{j,k}\varLambda_{i,k}+C_{j,k}$. The obtained Hessian matrix can be expressed as:
\begin{align}
H=
\begin{bmatrix}
\varphi_{1,1} & \varphi_{1,2} \\
\varphi_{2,1} & \varphi_{2,2} \label{46}
\end{bmatrix}
\end{align}
We can see that $\varphi_{1,1}$ and $\varphi_{2,2}$ in (\ref{56}) are the first order principle minors and also negative. Moreover, it can be evident that the second order minors are the determinant of (\ref{56}) and can be written as
\begin{align}
det H = \varphi_{1,1}\varphi_{2,2} - \varphi_{1,2}\varphi_{2,1} >0 \label{47}
\end{align}

\section*{Appendix D: Proof of Proposition 4}
The derivative of Equation 19 with respect to $\varLambda_{i,k}$ is 
\begin{align}
&\frac{\partial L(\boldsymbol{\varLambda_k},\boldsymbol{\lambda_k},\mu_k,\epsilon_k)}{\partial \varLambda_{i,k}}= \frac {A_{i,k}}{\ln(2)(A_{i,k}\varLambda_{i,k}+B_{i,k}\varLambda_{j,k}+C_{i,k})}\nonumber\\
&+\frac {A_{j,k}B_{j,k}\varLambda_{j,k}}{\ln(2)(B_{j,k}\varLambda_{i,k}+C_{j,k})(A_{j,k-}\varLambda_{j,k}+B_{j,k}\varLambda_{i,k}+C_{j,k})}\nonumber\\
&-D \label{48}
\end{align}
where $D =\varPi P_k-\lambda_{i,k}A_{i,k}+\lambda_{j,k}(2^{R_{min}}-1)B_{j,k}+\epsilon_k $. Put $\varLambda_{j,k}=1-\varLambda_{i,k}$ in Equation (34) which results in
\begin{align}
&\frac {A_{i,k}}{\ln(2)(X_{i,k}\varLambda_{i,k}+W_{i,k})}-\frac {\gamma^k_{j\rightarrow j}B_{j,k}}{\ln(2)(Y_{j,k}\varLambda_{i,k}+W_{j,k})}-D \label{49}
\end{align}
where $X_{i,k}=A_{i,k}-B_{i,k}$, $Y_{j,k}=B_{j,k}-A_{j,k}$, $W_{i,k}=B_{i,k}+C_{i,k}$ and $W_{j,k}=A_{j,k}+C_{j,k}$.\\ 
After some manipulation, Equation (35) results as:
\begin{align}
& A_{i,k}(Y_{j,k}\varLambda_{i,k}+W_{j,k})-{\gamma^k_{j\rightarrow j}B_{j,k}}(X_{i,k}\varLambda_{i,k}+W_{i,k})\nonumber\\
&-\ln(2)D(Y_{j,k}\varLambda_{i,k}+W_{j,k})(X_{i,k}\varLambda_{i,k}+W_{i,k})=0 \label{50}
\end{align}
After expanding and writing in $ax^2+bx+c$
\begin{align}
&(-\ln(2)DX_{i,k}Y_{j,k})\varLambda_{i,k}^2+(A_{i,k}Y_{j,k}-{\gamma^k_{j\rightarrow j}B_{j,k}}X_{i,k}\nonumber\\
&-\ln(2)DX_{i,k}W_{j,k}-\ln(2)DY_{j,k}W_{i,k})\varLambda_{i,k}+(A_{i,k}W_{j,k}\nonumber\\
&-{\gamma^k_{j\rightarrow j}B_{j,k}}W_{i,k}-\ln(2)DW_{i,k}W_{j,k}) \label{51}
\end{align}
The solution of above problem is as follow,
\begin{align}
\varLambda_{i,k}=\bigg[\frac {-b \pm \sqrt{b^2-4ac}}{2a}\bigg]^+ \label{52}
\end{align}
The proof is completed.

\bibliographystyle{IEEEtran}
\bibliography{Wali_Ref}

\begin{thebibliography}{10}
\providecommand{\url}[1]{#1}
\csname url@samestyle\endcsname
\providecommand{\newblock}{\relax}
\providecommand{\bibinfo}[2]{#2}
\providecommand{\BIBentrySTDinterwordspacing}{\spaceskip=0pt\relax}
\providecommand{\BIBentryALTinterwordstretchfactor}{4}
\providecommand{\BIBentryALTinterwordspacing}{\spaceskip=\fontdimen2\font plus
\BIBentryALTinterwordstretchfactor\fontdimen3\font minus
  \fontdimen4\font\relax}
\providecommand{\BIBforeignlanguage}[2]{{%
\expandafter\ifx\csname l@#1\endcsname\relax
\typeout{** WARNING: IEEEtran.bst: No hyphenation pattern has been}%
\typeout{** loaded for the language `#1'. Using the pattern for}%
\typeout{** the default language instead.}%
\else
\language=\csname l@#1\endcsname
\fi
#2}}
\providecommand{\BIBdecl}{\relax}
\BIBdecl

\bibitem{9261963}
F.~Jameel \emph{et~al.}, ``{NOMA-enabled backscatter communications: Toward
  battery-free IoT networks},'' \emph{IEEE Internet of Things Magazine},
  vol.~3, no.~4, pp. 95--101, Dec. 2020.

\bibitem{liu2019next}
W.~Liu \emph{et~al.}, ``{Next generation backscatter communication: systems,
  techniques, and applications},'' \emph{EURASIP J. Wireless Commun. Netw.},
  vol. 2019, no.~1, pp. 1--11, 2019.

\bibitem{9468352}
W.~U. Khan \emph{et~al.}, ``{NOMA-enabled optimization framework for
  next-generation small-cell IoV networks under imperfect SIC decoding},''
  \emph{IEEE Transactions on Intelligent Transportation Systems}, pp. 1--10,
  2021.

\bibitem{9516696}
A.~U. Khan \emph{et~al.}, ``{An enhanced spectrum reservation framework for
  heterogeneous users in CR-enabled IoT networks},'' \emph{IEEE Wireless
  Communications Letters}, pp. 1--1, 2021.

\bibitem{2021726}
A.~U. {Khan} \emph{et~al.}, ``{Spectrum utilization efficiency in CRNs with
  hybrid spectrum access and channel reservation: A comprehensive analysis
  under prioritized traffic},'' \emph{Future Generation Computer Systems}, vol.
  125, pp. 726--742, 2021.

\bibitem{8933559}
F.~{Jameel} \emph{et~al.}, ``{Applications of backscatter communications for
  healthcare networks},'' \emph{IEEE Network}, vol.~33, no.~6, pp. 50--57,
  Nov.-Dec. 2019.

\bibitem{8861078}
W.~U. {Khan} \emph{et~al.}, ``{Joint spectral and energy efficiency
  optimization for downlink NOMA networks},'' \emph{IEEE Trans. Cogn. Commun.
  Netw.}, vol.~6, no.~2, pp. 645--656, June 2020.

\bibitem{8253544}
X.~{Lu} \emph{et~al.}, ``{Ambient backscatter assisted wireless powered
  communications},'' \emph{IEEE Wireless Communications}, vol.~25, no.~2, pp.
  170--177, Apr. 2018.

\bibitem{khan2021energy}
W.~U. Khan \emph{et~al.}, ``{Energy-efficient resource allocation for 6G
  backscatter-enabled NOMA IoV networks},'' \emph{IEEE Transactions on
  Intelligent Transportation Systems}, pp. 1--11, 2021.

\bibitem{ihsan2021energy}
A.~Ihsan \emph{et~al.}, ``{Energy-efficient cackscatter aided uplink NOMA
  roadside sensor communications under channel estimation errors},''
  \emph{arXiv preprint arXiv:2109.05341}, 2021.

\bibitem{8368232}
N.~{Van Huynh} \emph{et~al.}, ``{Ambient backscatter communications: A
  contemporary survey},'' \emph{IEEE Commun. Surveys Tutor.}, vol.~20, no.~4,
  pp. 2889--2922, Fourthquarter 2018.

\bibitem{9154358}
O.~{Maraqa} \emph{et~al.}, ``{A survey of rate-optimal power domain NOMA with
  enabling technologies of future wireless networks},'' \emph{IEEE Commun.
  Surveys Tutor.}, vol.~22, no.~4, pp. 2192--2235, 2020.

\bibitem{9479745}
W.~U. Khan \emph{et~al.}, ``{Joint spectrum and energy optimization of
  NOMA-enabled small-cell networks with QoS guarantee},'' \emph{IEEE
  Transactions on Vehicular Technology}, vol.~70, no.~8, pp. 8337--8342, Aug.
  2021.

\bibitem{7842433}
Z.~{Ding} \emph{et~al.}, ``{Application of non-orthogonal multiple access in
  LTE and 5G networks},'' \emph{IEEE Commun. Mag.}, vol.~55, no.~2, pp.
  185--191, Feb 2017.

\bibitem{ali2022fair}
Z.~Ali \emph{et~al.}, ``{Fair power allocation in cooperative cognitive systems
  under NOMA transmission for future IoT networks},'' \emph{Alexandria
  Engineering Journal}, vol.~61, no.~1, pp. 575--583, 2022.

\bibitem{8692391}
H.~{Guo} \emph{et~al.}, ``{Cooperative ambient backscatter system: A symbiotic
  radio paradigm for passive IoT},'' \emph{IEEE Wireless Commun. Lett.},
  vol.~8, no.~4, pp. 1191--1194, Aug. 2019.

\bibitem{9051982}
Y.~{Ye} \emph{et~al.}, ``{On the outage performance of ambient backscatter
  communications},'' \emph{IEEE IoT J.}, vol.~7, no.~8, pp. 7265--7278, Aug.
  2020.

\bibitem{jameel2019simultaneous}
F.~Jameel \emph{et~al.}, ``{Simultaneous harvest-and-transmit ambient
  backscatter communications under Rayleigh fading},'' \emph{EURASIP J.
  Wireless Commun. Net.}, vol. 2019, no.~1, pp. 1--9, Dec. 2019.

\bibitem{8423609}
J.~{Qian} \emph{et~al.}, ``{IoT communications with $M$ -PSK modulated ambient
  backscatter: Algorithm, analysis, and implementation},'' \emph{IEEE IoT J.},
  vol.~6, no.~1, pp. 844--855, Feb. 2019.

\bibitem{8093703}
B.~{Lyu} \emph{et~al.}, ``{The optimal control policy for RF-powered
  backscatter communication networks},'' \emph{IEEE Trans. Vehicular Techn.},
  vol.~67, no.~3, pp. 2804--2808, Mar. 2018.

\bibitem{9129364}
F.~{Jameel} \emph{et~al.}, ``{Low latency ambient backscatter communications
  with deep Q-learning for beyond 5G applications},'' in \emph{IEEE
  VTC2020-Spring}, 2020, pp. 1--6.

\bibitem{9363336}
X.~{Li} \emph{et~al.}, ``{Physical layer security of cognitive ambient
  backscatter communications for green internet-of-things},'' \emph{IEEE Trans.
  Green Commun. Netw.}, pp. 1--1, 2021.

\bibitem{9024401}
F.~{Jameel} \emph{et~al.}, ``{Towards intelligent IoT networks: Reinforcement
  learning for reliable backscatter communications},'' in \emph{2019 IEEE
  Globecom Workshops (GC Wkshps)}, 2019, pp. 1--6.

\bibitem{9162720}
------, ``{Reinforcement learning for scalable and reliable power allocation in
  SDN-based backscatter heterogeneous network},'' in \emph{IEEE INFOCOM 2020 -
  IEEE Conference on Computer Communications Workshops (INFOCOM WKSHPS)}, 2020,
  pp. 1069--1074.

\bibitem{le2019outage}
C.-B. Le and D.-T. Do, ``{Outage performance of backscatter NOMA relaying
  systems equipping with multiple antennas},'' \emph{Electronics Letters},
  vol.~55, no.~19, pp. 1066--1067, Sept. 2019.

\bibitem{8636518}
Q.~{Zhang} \emph{et~al.}, ``{Backscatter-NOMA: A symbiotic system of cellular
  and internet-of-things networks},'' \emph{IEEE Access}, vol.~7, pp.
  20\,000--20\,013, 2019.

\bibitem{9131891}
X.~{Li} \emph{et~al.}, ``{Secrecy analysis of ambient backscatter NOMA systems
  under I/Q imbalance},'' \emph{IEEE Trans. Veh. Technol.}, vol.~69, no.~10,
  pp. 12\,286--12\,290, Oct. 2020.

\bibitem{9345447}
W.~U. {Khan} \emph{et~al.}, ``{Backscatter-enabled efficient V2X communication
  with non-orthogonal multiple access},'' \emph{IEEE Trans. Veh. Technol.},
  vol.~70, no.~2, pp. 1724--1735, Feb. 2021.

\bibitem{8962090}
Y.~{Liao} \emph{et~al.}, ``{Resource allocation in NOMA-enhanced full-duplex
  symbiotic radio networks},'' \emph{IEEE Access}, vol.~8, pp.
  22\,709--22\,720, 2020.

\bibitem{8439079}
J.~{Guo} \emph{et~al.}, ``{Design of non-orthogonal multiple access enhanced
  backscatter communication},'' \emph{IEEE Trans. Wireless Commun.}, vol.~17,
  no.~10, pp. 6837--6852, Oct. 2018.

\bibitem{8761125}
A.~{Farajzadeh} \emph{et~al.}, ``{UAV data collection over NOMA backscatter
  networks: UAV altitude and trajectory optimization},'' in \emph{IEEE ICC
  2019}, 2019, pp. 1--7.

\bibitem{8851217}
G.~{Yang} \emph{et~al.}, ``{Resource allocation in NOMA-enhanced backscatter
  communication networks for wireless powered IoT},'' \emph{IEEE Wireless
  Commun. Lett.}, vol.~9, no.~1, pp. 117--120, Jan. 2020.

\bibitem{8877102}
S.~{Zeb} \emph{et~al.}, ``{NOMA enhanced backscatter communication for green
  IoT networks},'' in \emph{IEEE 16th International Symposium on Wireless
  Communication Systems (ISWCS)}, 2019, pp. 640--644.

\bibitem{li2019secure}
Y.~Li \emph{et~al.}, ``{Secure beamforming in MISO NOMA backscatter device
  aided symbiotic radio networks},'' \emph{arXiv preprint arXiv:1906.03410},
  2019.

\bibitem{9223730}
Y.~{Xu} \emph{et~al.}, ``{Energy efficiency maximization in NOMA enabled
  backscatter communications with QoS guarantee},'' \emph{IEEE Wireless Commun.
  Lett.}, vol.~10, no.~2, pp. 353--357, Feb. 2021.

\bibitem{9319204}
X.~{Li} \emph{et~al.}, ``{Hardware impaired ambient backscatter NOMA systems:
  Reliability and security},'' \emph{IEEE Trans. Commun.}, pp. 1--1, 2021.

\bibitem{9328505}
W.~U. {Khan} \emph{et~al.}, ``{Backscatter-enabled NOMA for future 6G systems:
  A new optimization framework under imperfect SIC},'' \emph{IEEE Commun.
  Lett.}, pp. 1--1, 2021.

\bibitem{9219112}
X.~{Li} \emph{et~al.}, ``{Cooperative wireless-powered NOMA relaying for B5G
  IoT networks with hardware impairments and channel estimation errors},''
  \emph{IEEE Internet Things J.}, pp. 1--1, 2020.

\bibitem{9261140}
W.~U. {Khan} \emph{et~al.}, ``{Spectral efficiency optimization for next
  generation NOMA-enabled IoT networks},'' \emph{IEEE Trans. Veh. Technol.},
  vol.~69, no.~12, pp. 15\,284--15\,297, Dec. 2020.

\bibitem{8540884}
M.~R. Zamani \emph{et~al.}, ``{Energy-efficient power allocation for NOMA with
  imperfect CSI},'' \emph{IEEE Trans. Veh. Technol.}, vol.~68, no.~1, pp.
  1009--1013, Jan. 2019.

\bibitem{9134383}
F.~Jameel \emph{et~al.}, ``{Efficient power-splitting and resource allocation
  for cellular V2X communications},'' \emph{IEEE Transactions on Intelligent
  Transportation Systems}, vol.~22, no.~6, pp. 3547--3556, June 2021.

\end{thebibliography}

\end{document}